\documentclass{article}
\usepackage{amsfonts,psfrag,epsfig}
\usepackage{amsfonts,epsfig}
\usepackage{color}
\usepackage{amsmath,amssymb,hyperref}
\usepackage[amsmath,thmmarks]{ntheorem}
\usepackage{tabularx}
\usepackage{epstopdf}

\usepackage[margin = 1.1in]{geometry}

\hypersetup{colorlinks=true}

\newtheorem{theorem}{Theorem}
\newtheorem{lemma}[theorem]{Lemma}

{ \theoremstyle{nonumberplain}\theoremsymbol{\ensuremath{\Box}}
\theorembodyfont{\normalfont}
\newtheorem{proof}{Proof.}
}
\usepackage{cite}
\usepackage{algorithmic}
\usepackage{array}
\usepackage[caption=false,font=normalsize,labelfont=sf,textfont=sf]{subfig}

\begin{document}
\title{Maximum Weight Matching using Odd-sized Cycles: Max-Product Belief Propagation and Half-Integrality}

\author{Sungsoo~Ahn,~
Michael~Chertkov,~
Andrew~E.~Gelfand,~
Sejun~Park 
and~Jinwoo~Shin
\thanks{S.\ Ahn, S.\ Park and J.\ Shin are with the School of Electrical Engineering, Korea Advanced Institute of Science and Technology (KAIST), South Korea. Email: \{sungsoo.ahn, sejun.park, jinwoos\}@kaist.ac.kr. %
M.\ Chertkov is with Theoretical Division \& Center for Nonlinear Studies,
Los Alamos National Laboratory, Los Alamos, NM 87545, USA and
Skolkovo Institute of Science and Technology, 143026 Moscow, Russia. Email: chertkov@lanl.gov. %
A.\ E.\ Gelfand is with Engineers Gate Manager LP
767 5th Ave
New York, NY 10153, 
USA. Email: andrew.gelfand@gmail.com.}%
}
\maketitle

\begin{abstract}
We study the Maximum Weight Matching (MWM) problem for general graphs through 
the max-product Belief Propagation (BP) and related Linear Programming (LP).
The BP approach
provides distributed heuristics for finding the Maximum A Posteriori (MAP) assignment in a joint probability distribution represented by a Graphical Model (GM) and respective LPs
can be considered as continuous relaxations of the discrete MAP problem. 
It was recently shown that a BP algorithm
converges to the correct MAP/MWM assignment under a simple GM formulation of MWM as long as the corresponding LP relaxation is tight.
First, under the motivation for forcing the tightness condition,
we consider a new GM formulation of MWM, say C-GM, using non-intersecting odd-sized cycles in the graph:
the new corresponding LP relaxation, say C-LP, becomes tight for more MWM instances.
However, 
the tightness of C-LP now does not guarantee 
such convergence and correctness of 
the new BP on C-GM.
To address the issue, we introduce  
a novel graph transformation applied to C-GM, which results in
another GM formulation of MWM, and
prove that the respective BP on it converges to the correct MAP/MWM assignment as long as C-LP is tight.
Finally, we also show that C-LP always has half-integral solutions,
which leads to
an efficient BP-based MWM heuristic consisting of making sequential, ``cutting plane'', modifications to the underlying GM. Our experiments show that this BP-based cutting plane heuristic performs as well as that based on traditional LP solvers.
\end{abstract}


%

\section{Introduction}
%
%
%
%
{Graphical} Models (GMs) have been utilized 
for reasoning in a variety of scientific fields \cite{05YFW,08RU,09MM,08WJ}. Such models use a graph structure to encode the joint probability distribution, where vertices correspond to random variables and edges specify conditional dependencies. An important inference task in many applications involving GMs is to find the most likely assignment to the variables in a GM - the maximum a posteriori (MAP) configuration. The max-product Belief Propagation (BP) is a popular approach for approximately solving the MAP inference problem. BP is an iterative, message-passing algorithm that is exact on tree structured GMs. However, BP often shows remarkably strong heuristic performance beyond trees, i.e., on GMs with loops. Distributed implementation, associated ease of programming and strong parallelization potential are the main reasons for the growing popularity of the BP algorithm, e.g., see \cite{09GLG,10LGKBGH} for recent discussions of BP's parallel implementations.

The convergence and correctness of BP was recently established for a certain class of loopy GM formulations of several classical combinatorial optimization problems, including matchings \cite{08BSS,11SMW,07HJ}, perfect matchings \cite{11BBCZ}, shortest paths \cite{08RT}, independent sets \cite{07SSW} and network flows \cite{10GSY}. The important common feature of these instances is that BP converges {in polynomial-time} to a correct MAP assignment when the Linear Programming (LP) relaxation of the MAP inference problem is tight, i.e., when it shows no integrality gap. While this demonstrates that LP tightness is necessary for the convergence and correctness of BP, it is unfortunately not sufficient in general. In other words, BP may not work even when the corresponding LP relaxation to the MAP inference problem is tight. 

To handle this issue, several message-passing alternatives to BP have been  studied, e.g.,
TRW (tree-reweighted) \cite{WJW05,KW05}, MPLP (max-product linear programming) \cite{GJ08}, MSD (max-sum diffusion) \cite{W07},
TCBO (tree-consistency bound optimization) \cite{MGW09}.
When applied to binary pairwise GMs, these algorithms are provably convergent (or conjectured to be convergent)
while TRW and MPLP are guaranteed to converge to the solution of 
the respective LP relaxation.
However, these convergence guarantees come from the monotone decreasing property of a certain global objective with respect to the number of iterations, thus making analysis of the convergence rates difficult.
Furthermore, in the case of non-pairwise GMs, none of the aforementioned methods guarantee correctness and convergence simultaneously even in the case when
the LP relaxation is tight.
We remark that one can try generic centralized optimizing approaches of
gradient-descent, simplex, ellipsoid or interior-point types
for solving the LP relaxation, but this may not be practical for large-scale problems.

In this paper we restrict our attention to a specific class of {non-pairwise} GMs corresponding to the Maximum Weight Matching (MWM) problem. Under this model, we address the question
whether BP can be used as an iterative, message-passing-based LP solver. 
Our work builds upon that of Sanghavi, Malioutov and Willsky \cite{11SMW}, who studied BP under
a simple GM formulation of the MWM problem on an arbitrary graph. The authors showed that 
BP on the simple GM converges to the correct, MAP solution 
if the corresponding LP relaxation is tight. 
Unfortunately, the tightness is not guaranteed in general. 
To enforce the tightness we add some extra constraints (thus utilizing the classic cutting plane method \cite{54DFJ}), i.e., odd-sized cycle inequalities, to the simple GM and LP, which
results in new GM, say C-GM and the corresponding LP relaxation, say C-LP.
Note that similar cycle constraints have been studied for tightening LPs of generic MAP problems \cite{SCL12}.
However, the addition of these extra constraints unfortunately result in a violation of 
BP's
convergence and correctness on C-GM even under the tightness of the corresponding LP relaxations.

To resolve this issue, we propose
a novel graph transformation applied to C-GM
and show that the `modified' max-product BP
converges to the MWM assignment as long as C-LP is tight. 
For the guarantee, the only restriction placed on C-GM 
is that the set of cycle constraints 
remains non-intersecting (in edges).
In other words, we design a new GM, say C-GM$^{\prime}$, 
via a certain graphical transformation of C-GM, i.e., collapsing odd cycles into new vertices and defining new weights on the contracted graph, so that
the BP algorithm on C-GM$^{\prime}$ converges to MAP under the tightness of C-LP.
It is important to emphasize that the MAP assignments of C-GM and C-GM$^{\prime}$ are in one-to-one correspondence. 
This clever graphical transformation allows us to use the computational tree technique similar to \cite{11SMW}
(but our case is much more complex)
for guaranteeing the BP performance over C-GM$^{\prime}$,
which is impossible under C-GM.

In summary, we construct a novel GM so that the corresponding max-product BP on it converges to the MWM assignment {in polynomial-time} as long as the MWM-LP relaxation with cycle constraints is tight. Furthermore, we prove that the MWM-LP relaxation has half-integral solutions in general. 
Combining these theoretical results naturally guides a cutting-plane method suggesting a sequential and adaptive
design of GMs using respective BPs. We use the output of BP to identify odd-sized cycle constraints - ``cuts'' - to add to the MWM-LP, construct a new GM using our graphical transformation, run BP and repeat. We evaluate this heuristic approach empirically and show that its performance is close to the traditional cutting-plane approach employing LP solvers rather than BP, i.e., the distributed BP is as powerful as the LP solver admitting only global implementation.
To our knowledge, our result is the first one to suggest how to ``fix" BP via a graph transformation so that it works properly, i.e., succeeds in recovering the desired LP solution. 
We believe that our success in crafting a graphical transformation will offer useful insights into the design and analysis of BP algorithms on a wider class of MAP inference problems.

\subsection{Organization}

In Section \ref{sec:pre}, we introduce a standard GM formulation of the MWM problem as well as the corresponding BP
and LP. In Section \ref{sec:gtBP}, we introduce our new GM and describe performance guarantees of the respective BP algorithm. In Section \ref{sec:cutting}, we describe a cutting-plane(-like) method using BP for the MWM problem
as well as a half-integrality proof for the MWM-LP relaxation using odd-sized cycles.
In particular, we show its empirical performance for random MWM instances. Section \ref{sec:concl} is reserved for brief conclusions and discussing the path forward.

\section{Preliminaries}\label{sec:pre}

\subsection{Graphical Model for Maximum Weight Matchings}

A joint distribution of $n$ (discrete) random variables $Z=[Z_i]\in \Omega^n$ is called a Graphical Model (GM) if it factorizes as follows: 
\begin{equation}
	\Pr[Z=z]~\propto~\prod_{\alpha\in F} \psi_{\alpha} (z_\alpha),\qquad\mbox{for}\quad 
	z=[z_i]\in \Omega^n,\label{eq:generic_gm}
\end{equation}
where $F$ is a collection of subsets of $\Omega$, $z_\alpha=[z_i:i\in \alpha \subset \Omega]$ is a subset of variables, and $\psi_{\alpha}$ is some (given) non-negative function. The function $\psi_\alpha$ is called a factor (variable) function if $|\alpha|\geq 2$ ($|\alpha|=1$). For variable functions $\psi_\alpha$ with $\alpha=\{i\}$, one simply writes $\psi_\alpha = \psi_i$. One calls $z$ a valid assignment if $\Pr[Z=z]>0$. The MAP assignment $z^*$ is defined as
\begin{equation*}
z^*~=~\arg\max_{z\in \Omega^{n}} \Pr[Z=z].
\end{equation*}
	
Let us introduce the Maximum Weight Matching (MWM) problem and its related GM. Suppose we are given an undirected graph $G=(V,E)$ with weights $ \{w_e:e\in E\}$ assigned to its edges. A \emph{matching} is a set of edges without common vertices. The weight of a matching is the sum of corresponding edge weights. The MWM problem consists of finding a matching of maximum weight. Associate a binary random variable with each edge $X=[X_e]\in \{0,1\}^{|E|}$ and consider the probability distribution defined as
\begin{equation*}
	\mbox{C-GM}:\quad\Pr[X=x] ~\propto~ \prod_{e\in E}e^{w_ex_e} \prod_{i\in V} \psi_i(x) \prod_{C\in\mathcal C} \psi_C(x),
\end{equation*}
for
$x=[x_e]\in \{0,1\}^{|E|},
$
where $\psi_i(x), \psi_C(x)$ are
\begin{align*}
&\psi_i(x)=\begin{cases}
1 &\mbox{if}~\sum_{e\in \delta(i)} x_e \leq 1\\
0&\mbox{otherwise}
\end{cases}\\
&\psi_C(x)=
\begin{cases}
1 &\mbox{if}~\sum_{e\in E(C)} x_e\leq \frac{|C|-1}2 \\
0 &\mbox{otherwise}
\end{cases}.
\end{align*}
Here $\mathcal{C}$ is a set of odd-sized cycles $\mathcal C\subset 2^V$, $\delta(i)=\{(i,j)\in E\}$ and $E(C)=\{(i,j)\in E:i,j\in C\}$. Throughout the paper, we assume that cycles are non-intersecting in edges, i.e., $E(C_1)\cap E(C_2)=\emptyset$ for all $C_1,C_2\in \mathcal C.$
It is easy to see that a MAP assignment $x^*$ for C-GM induces a MWM in $G$,
{and the factor $\psi_C$ is redundant, i.e.,
it does not change the distribution C-GM. However, as explained in the next section,  the corresponding BP algorithm
and related LP relaxation depend on the choice of $\mathcal C$.}

\subsection{Belief Propagation and Linear Programming for Maximum Weight Matchings}\label{sec:bp}

In this section, we introduce max-product Belief Propagation (BP)
and the Linear Programming (LP) relaxation to computing the MAP assignment in C-GM.
We first describe the BP algorithm for the general GM \eqref{eq:generic_gm}, then tailor the algorithm
to C-GM.
The BP algorithm updates the set of messages $\{m^{t}_{\alpha\rightarrow i}(z_i),m^{t}_{i\rightarrow\alpha}(z_i):z_i\in\Omega\}$
between each variable $i$ and its associated factors $\alpha\in F_i=\{\alpha\in F:i\in \alpha, |\alpha|\geq 2\}$ using the following update rules:
\begin{align*}
&m^{t+1}_{\alpha\rightarrow i}(z_i) ~=~ 
\max_{z^\prime:z^\prime_i=z_i} \psi_\alpha (z^\prime) \prod_{j\in \alpha\setminus i} m_{j\rightarrow \alpha}^t (z^\prime_j)\\
&m^{t+1}_{i\rightarrow\alpha}(z_i) ~=~  \psi_i(z_i)\prod_{\alpha^{\prime}\in F_i\setminus \alpha} m_{\alpha^{\prime} \rightarrow i}^t (z_i).
\end{align*}

Here $t$ denotes time 
and initially
	$m^0_{\alpha\to i}(\cdot)=m^0_{i\to\alpha}(\cdot)=1$.
Given a set of messages $\{m_{i\to\alpha}(\cdot),m_{\alpha\to i}(\cdot))\}$, the BP (max-marginal) beliefs $\{n_i(z_i)\}$ are defined as follows:
\begin{eqnarray*}
n_i(z_i)~=~\psi_i(z_i)\prod_{\alpha\in F_i} m_{\alpha\to i}(z_i).
\end{eqnarray*}
For C-GM, one introduces
$n^t_{e}(\cdot)$ to denote the BP belief on edge $e\in E$ at time $t$.
The algorithm outputs MAP estimate at time $t$, $x^{\mbox{\tiny BP}}(t)=[x^{\mbox{\tiny BP}}_e(t)]\in \left[0,?,1\right]^{|E|}$, where
\begin{equation*}
x^{\mbox{\tiny BP}}_e(t)=\begin{cases}
1&\mbox{if}~n^t_{e}(0)<n^t_{e}(1)\\
?&\mbox{if}~n^t_{ij}(0)=n^t_{e}(1)\\
0&\mbox{if}~n^t_{e}(0)>n^t_{e}(1)
\end{cases}.
\end{equation*}
The LP relaxation of the MAP expression for C-GM is
\begin{align*}
\mbox{C-LP}:\quad&\max~ \sum_{e\in E} w_e x_e\\
\mbox{s.t.}\quad&\sum_{e\in \delta(i)} x_e\leq 1,\quad\forall i\in V,\\	
&\sum_{e\in E(C)} x_e\leq \frac{|C|-1}2,\quad\forall~ C\in\mathcal C,\\ 
&x_e\in[0,1].
\end{align*}
Observe that for any $\mathcal C$, if the solution $x^{\mbox{\tiny C-LP}}$ to C-LP is integral, i.e., $x^{\mbox{\tiny C-LP}}\in \{0,1\}^{|E|}$, then it is a MAP assignment, i.e., $x^{\mbox{\tiny C-LP}} = x^*$. {The role of additional cycle inequalities involving $\mathcal C$ is forcing C-LP to become tight, i.e., to show an integral solution. It is important to stress that the cycle inequalities just introduced are 
different from (much more popular) the Edmonds' blossom inequalities \cite{65Edm} defined for an odd-sized subset $S\subset V$,
\begin{equation*}
\sum_{e\in \delta(S)} x_e \geq 1,
\end{equation*}
where $\delta(S)$ is the set of edges between $S$ and $V\setminus S$. We consider the cycle inequalities instead because it makes the BP analysis tractable, but also (and mainly) because
it leads to a simple cutting-plane scheme related to BP and that takes advantage of the BP's half-integral property
(see Section \ref{sec:cutting}). This half-integrality plays a key role in designing  
our BP-based cutting-plane scheme and it does not hold in general for the Edmonds' inequalities. 
We also note that a polynomial-time cutting-plane algorithm using the Edmonds' inequalities was recently 
established \cite{12CLS}. However, the method of \cite{12CLS} is rather complex, e.g. in focusing on careful selection of blossoms which do not brake half-integrality, and thus dependent on 
 traditional, centralized methods (e.g., ellipsoid, interior-point or simplex) that are generally not practical for large-scale problems.}

Sanghavi, Malioutov and Willsky \cite{11SMW} 
proved the following theorem connecting the performance of BP and C-LP in a special case:
\begin{theorem}\label{thm:sujay}
	If $\mathcal C=\emptyset$ and the solution of C-LP is integral and unique, then $x^{\mbox{\tiny BP}}(t)$
	under C-GM converges
	to the MWM assignment $x^*$ in $O(w_{\max}/c)$ iterations, where $w_{\max},c$ are defined as
	\begin{align*}
	&w_{\max}:=\max_{e\in E} w_e\\
	&c:=\inf_{x\neq x^*: \mbox{\tiny $x$ is feasible to C-LP}}\frac{w \cdot(x^*-x)}{|x^*-x|}>0.
	\end{align*}
\end{theorem}
Adding small random noise to every weight guarantees the uniqueness condition required by Theorem \ref{thm:sujay}.
{In particular, one can design random noise utilizing the Isolation Lemma \cite{MVV87} thus
guaranteeing a polynomial convergence (with respect to $|V|$) of BP independent of $c$ through technical arguments typical for the Lemma, see e.g. \cite{10GSY}.}
One naturally makes a conjecture that Theorem \ref{thm:sujay} extends to a non-empty $\mathcal C$ since adding more cycles 
can help to reduce the integrality gap of C-LP.
However, the theorem does not hold when $\mathcal C\neq \emptyset$. For example, it is straightforward to check that BP does not converge for a triangle graph with edge weights $\{2,1,1\}$ and $\mathcal C$ consisting of the single cycle. This is true even though the solution of the corresponding C-LP is unique and integral.

\section{Graphical Transformation for Convergent and Correct Belief Propagation}\label{sec:gtBP}

The loss of convergence and correctness of BP when the MWM LP is tight (and unique),
however $\mathcal C\neq \emptyset$, motivates consideration of  this section.
We resolve the issue by designing a new GM, equivalent to the original GM, however such that when BP is applied to this new GM it converges to the MAP/MWM assignment whenever the LP relaxation is tight and unique - even if $\mathcal{C}\neq\emptyset$.
The new GM is defined on an auxiliary graph $G^{\prime} = (V^{\prime}, E^{\prime})$ with new weights, $\{w^{\prime}_e:e\in E^{\prime}\}$, as follows:
\begin{align*}
V^{\prime} &= V\cup \{i_C:C\in \mathcal C\}, 
\\
E^{\prime} &= E\cup\{(i_C,j):j\in V(C), C\in \mathcal C\}\setminus \{e:e\in \cup_{C\in \mathcal C}E(C)\},\\
w_e^{\prime} &=
\begin{cases}
\frac12\sum_{e^{\prime}\in E(C)} (-1)^{d_C(j,e^{\prime})} w_{e^\prime}
\\
\qquad\qquad\quad\mbox{if}~e=(i_C,j) ~\mbox{for some}~C\in \mathcal C\\
w_e\qquad\qquad\mbox{otherwise}
\end{cases}.
\end{align*}
Here $d_C(j,e)$ is the graph distance of $j$ and $e$ in cycle $C=(j_1,j_2,\dots,j_k)$, e.g., if $e=(j_2,j_3)$, then $d_C(j_1,e)=1$.
\begin{figure}[t!]
\begin{center}
\includegraphics[width=0.4\textwidth]{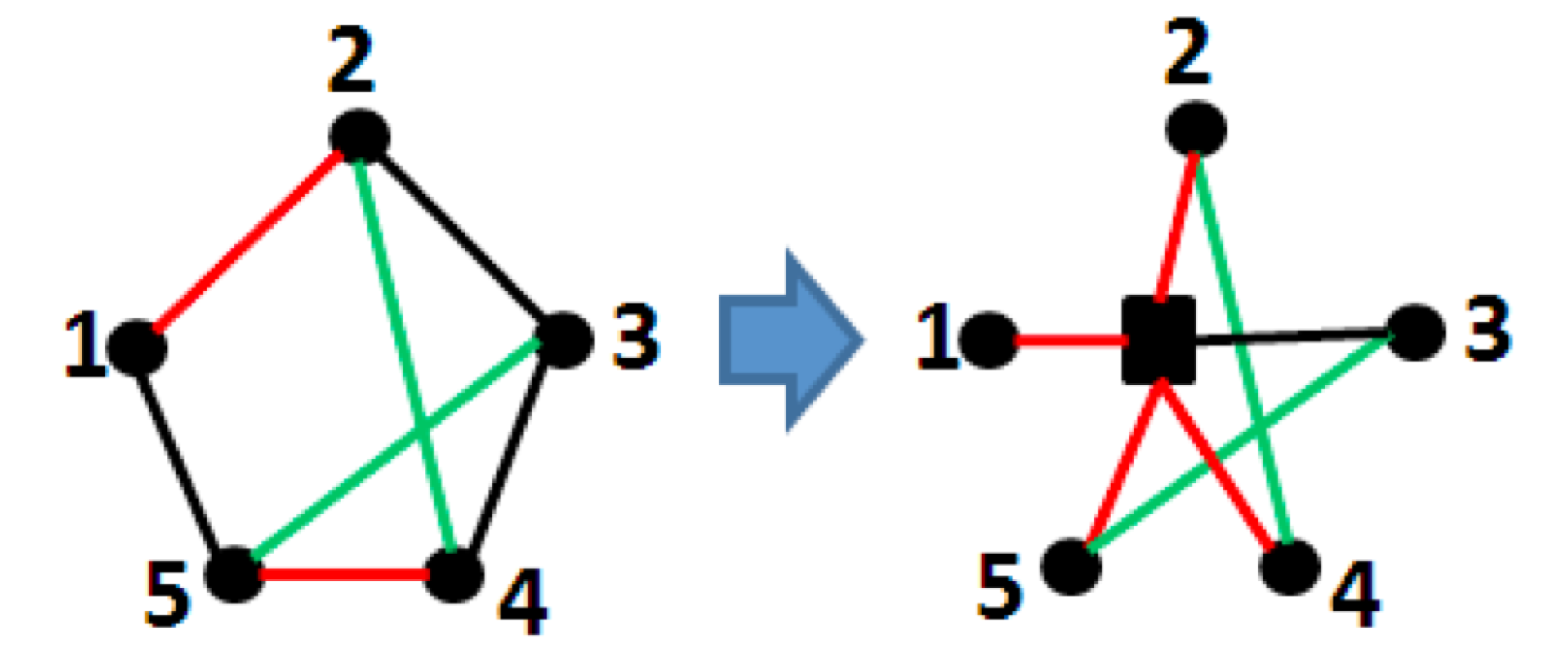}
\end{center}
\caption{\small Example of original graph $G$ (left) with an odd cycle of length 5 and new graph $G^{\prime}$ (right).}
\label{fig:graph_transformation}
\end{figure}
Associate binary variables for each new edge 
and consider the new probability distribution defined as
\begin{equation*}
{\mbox{C-GM}^\prime:}\quad\Pr[Y=y] ~\propto~ \prod_{e\in E^{\prime}}e^{w_e^{\prime}y_e}
\prod_{i\in V} \psi_i(y)\prod_{C\in \mathcal C} \psi_C(y),
\end{equation*}
for $
y=[y_e:e\in E^{\prime}]\in \{0,1\}^{|E^{\prime}|}$, where
\begin{align*}
\psi_i(y)&=\begin{cases}
1 &\mbox{if}~\sum\limits_{e\in \delta(i)} y_e \leq 1\\
0&\mbox{otherwise}
\end{cases},\\
\psi_C(y)&=
\begin{cases}
0 &\mbox{if}~\sum\limits_{e\in\delta(i_C)} y_e >|C|-1 \\
0 &\mbox{if}~\sum\limits_{j\in V(C)} (-1)^{d_C(j,e)} y_{i_C,j} \notin \{0,2\}\\
&\qquad\mbox{for some}~e\in E(C) \\
1 &\mbox{otherwise}
\end{cases}.
\end{align*}
It is not hard to check that the number of operations required to update messages at each round of BP under the above GM is $O(|V||E|)$,
since the number of non-intersecting cycles is at most $|E|$ and
message updates involving the factor $\psi_C$ can be efficiently done using dynamic programming.
We are now ready to state the main result of this paper.
\begin{theorem}\label{thm:main}
	If the solution of C-LP is integral and unique, then the BP-MAP estimate $y^{\mbox{\tiny BP}}(t)$
	applied to C-GM$^\prime$ converges
	to the MAP assignment $y^*$. Furthermore, the MWM assignment $x^*$ 
	is related to $y^*$ as follows:
	\begin{equation}\label{eq:mainthm}
		x_e^* =
	\begin{cases}
	\frac12\sum_{j\in V(C)} (-1)^{d_C(j,e)} y_{i_C,j}^*&\mbox{\rm if}~~e\in \bigcup_{C\in\mathcal C} E(C)\\
	\qquad\quad y_e^*&\mbox{\rm otherwise}
	\end{cases}.\end{equation}
\end{theorem}

The proof of Theorem \ref{thm:main} is provided in the following sections.
We also establish the convergence rate of the BP algorithm over C-GM$^\prime$
(see Lemma \ref{thm3}), {i.e., BP converges in $O(w_{\max}^\prime/c)$ where the constant $c$
is defined in Theorem \ref{thm:sujay}. As we mentioned earlier, one can 
use the Isolation Lemma \cite{MVV87} and related machinery
to establish a polynomial convergence (with respect to $|V|$)
independent of $c$.}
Let us emphasize that C-GM$^\prime$ is designed so that
each variable is associated to at most two factor nodes. We call this condition, which did not
hold for the original C-GM, the `degree-two' (DT) condition. 
The DT condition is necessary for analysis of BP using the computational tree technique, developed and advanced in  \cite{08BSS,11SMW,10GSY,63G,00FK,01WF, PS145}.
It will also
play a critical role in the proof of Theorem~\ref{thm:main}. We further remark that even under the DT condition
and given tightness/uniqueness of the LP, proving correctness and convergence of BP is still highly
non trivial. In our case, it requires careful study of the computation tree induced by BP with appropriate
truncation at the leaves.

\subsection{Main Lemma required to prove Theorem \ref{thm:main}}

Let us introduce the following auxiliary LP over the new graph and weights.
\begin{align}
\mbox{C-LP}^{\prime}&:~\max~ \sum_{e\in E^{\prime}} w_e^{\prime} y_e\notag\\
\mbox{s.t.}~&\sum_{e\in \delta(i)} y_e\leq 1,\quad\forall i\in V,  \quad y_e\in[0,1],\quad\forall e\in E^{\prime},
\label{eq1:clp'}\\
\quad\quad &\sum_{j\in V(C)} (-1)^{d_C(j,e)} y_{i_C,j}\in[0,2],\quad\forall e\in E(C),\\
&\sum_{e\in \delta(i_C)} y_e \leq {|C|-1},\quad\forall C\in \mathcal C.\label{eq2:clp'}
\end{align}
Next, consider the following one-to-one linear map between $x=[x_e:e\in E]$ and
$y=[y_e:e\in E^{\prime}]$:
\begin{align*}
y_e &=
\begin{cases}
\sum_{e^{\prime}\in E(C) \cap \delta(i)} x_{e^{\prime}}&\mbox{if}~e=(i,i_C)\\
\;\quad x_e&\mbox{otherwise}
\end{cases},\\
x_e &=
\begin{cases}
\frac12\sum_{j\in V(C)} (-1)^{d_C(j,e)} y_{i_C,j}&\mbox{if}~e\in \bigcup_{C\in\mathcal C}E(C)\\
\qquad\quad y_e&\mbox{otherwise}
\end{cases}.
\end{align*}
One can verify that under this maps C-LP = C-LP$^{\prime}$. Besides,
if solution $x^{\mbox{\tiny C-LP}}$ of C-LP is unique and integral then
solution $y^{\mbox{\tiny C-LP}^{\prime}}$ of C-LP$^{\prime}$ is also unique and integral, i.e.,
$y^{\mbox{\tiny C-LP}^{\prime}}=y^*$. Hence, \eqref{eq:mainthm} in Theorem \ref{thm:main} follows.
Furthermore, since $y^*=[y^*_e]$ is unique and integral,
there exists {$c^\prime = \Omega(c)$} such that
$$c^\prime:=\inf_{y\neq y^*: y~\mbox{\tiny feasible to C-LP}^{\prime}}\frac{w^{\prime}\cdot(y^*-y)}{|y^*-y|},$$
where $w^{\prime}=[w^{\prime}_e]$. These observations help us to establish the following lemma characterizing  performance of the max-product BP over C-GM$^\prime$.
\begin{lemma}\label{thm3}
	If solution $y^{\mbox{\tiny C-LP}^{\prime}}$ of C-LP$^{\prime}$ is integral and unique, i.e., $y^{\mbox{\tiny C-LP}^{\prime}}=y^*$, then
\begin{itemize}
\item if $y^*_e=1$, $n^t_e[1]>n^t_e[0]~$ for all $t>O\left(\frac{w^{\prime}_{\max}}{c^\prime}\right)$,
\item if $y^*_e=0$, $n^t_e[1]<n^t_e[0]~$ for all $t>O\left(\frac{w^{\prime}_{\max}}{c^\prime}\right)$,
\end{itemize}
where $n^t_e[\cdot]$ denotes the BP belief of edge $e$ at time $t$ under C-GM$^\prime$ and 
$w_{\max}^{\prime}=\max_{e\in E^{\prime}} \left|w^{\prime}_e\right|$.
\end{lemma}
Theorem \ref{thm:main} follows from this lemma directly.

\subsection{Proof of Lemma \ref{thm3}}

This section provides the complete proof of Lemma \ref{thm3} by contradiction.
Here we focus on the case of $y_e^{*}=1$, noticing that extending the proof arguments to the case of $y_e^*=0$ is straightforward. To derive a contradiction, let us assume that $n^t_e[1]\leq n^t_e[0]$ and
construct a tree-structured GM $T_e(t)$ of depth $t+1$, also known as
the computational tree of BP \cite{97W}, using the following scheme.
\begin{itemize}
	\item[1.] Add a copy of $Y_e\in\{0,1\}$ as the (root) variable (with variable function $e^{w_e^{\prime} Y_e}$).
	\item[2.] Repeat the following $t$ times for each leaf variable $Y_e$ for the current tree-structured GM.
	\begin{itemize}
		\item[2-1.] For each $i\in V$ such that $e\in \delta(i)$ and $\psi_i$ not associated with $Y_e$ of the current model, add $\psi_i$ as a factor (function) with
		copies of $\{Y_{e^{\prime}}\in\{0,1\}:e^{\prime}\in \delta(i)\setminus e\}$
	as child variables (with corresponding variable functions, i.e., $e^{w_{e^{\prime}}^{\prime} Y_{e^{\prime}}}$).
	    \item[2-2.] For each $C\in \mathcal C$ such that $e\in \delta(i_C)$ and $\psi_C$ not associated with $Y_e$ of the current model, add $\psi_C$ as a factor (function) with copies of
	$\{Y_{e^{\prime}}\in\{0,1\}:e^{\prime}\in \delta(i_C)\setminus e\}$
as the child variables (with corresponding variable functions, i.e., $e^{w_{e^{\prime}}^{\prime} Y_{e^{\prime}}}$).
\end{itemize}
\end{itemize}
It is  known \cite{97W} that there exists a MAP configuration $y^{\mbox{\tiny TMAP}}$ of $T_e(t)$ with $y^{\mbox{\tiny TMAP}}_e$ fixed to $0$ at the root variable.
Our goal is to find a new 
assignment $y^{\mbox{\tiny NEW}}$ on the computational tree $T_e(t)$
such that
$w^\prime \cdot y^{\mbox{\tiny NEW}} > w^\prime \cdot y^{\mbox{\tiny TMAP}}$.
This contradicts to the definition of $y^{\mbox{\tiny TMAP}}$ being a MAP configuration
and completes the proof of
Lemma \ref{thm3}.
In particular, we construct the following 
new assignment $y^{\mbox{\tiny NEW}}$ on the computational tree $T_e(t)$.

\begin{enumerate}
	\item[1.] Initially, set $y^{\mbox{\tiny NEW}} \leftarrow y^{\mbox{\tiny TMAP}}$ and denote $e$ as the root of the tree.
	\item[2.] Update $y^{\mbox{\tiny NEW}} \leftarrow {\tt FLIP}_e(y^{\mbox{\tiny NEW}})$, where ${\tt FLIP}_e(y)$ is the 0-1 (i.e., binary) vector made by flipping (i.e., changing from $0$ to $1$ or $1$ to $0$) the $e$'s position in $y$. 
	\item[3.] For the (unique) child factor $\psi$ associated with $e$, perform the following updates on
	$y^{\mbox{\tiny NEW}}$, where
	we say $\psi$ is satisfied by some assignment $y$ if $\psi(y)=1$.
	\begin{enumerate}
\item[3-1.] 
If $\psi$ is satisfied by both $y^{\mbox{\tiny NEW}}$ and ${\tt FLIP}_e(y^{*})$, then do nothing.
	    \item[3-2.] Else if there exists a $e$'s child $e^{\prime}$ through factor $\psi$ such that $y^{\mbox{\tiny NEW}}_{e^{\prime}}\neq y^*_{e^{\prime}}$ and
$\psi$ is satisfied by both ${\tt FLIP}_{e^{\prime}}(y^{\mbox{\tiny NEW}})$ and ${\tt FLIP}_{e^{\prime}}({\tt FLIP}_{e}(y^{*}))$, then
go to Step 2 with $e\leftarrow e^{\prime}$.
\item[3-3.] Otherwise, report ERROR.
	\end{enumerate}
\end{enumerate}
We note that 
there exists exactly one child factor $\psi$ in Step 3,
while we choose one child $e^{\prime}$ in Step 3-2 among many possible candidates. Due to this reason, the flip operations induce a path structure $P$ in tree $T_e(t)$.\footnote{$P$ may not have an alternating structure, i.e.,
both $y^{\mbox{\tiny NEW}}_e$ and its child $y^{\mbox{\tiny NEW}}_{e^{\prime}}$ can be flipped in a same way.}
\begin{figure*}[ht]
\begin{center}
\includegraphics[width=1\textwidth]{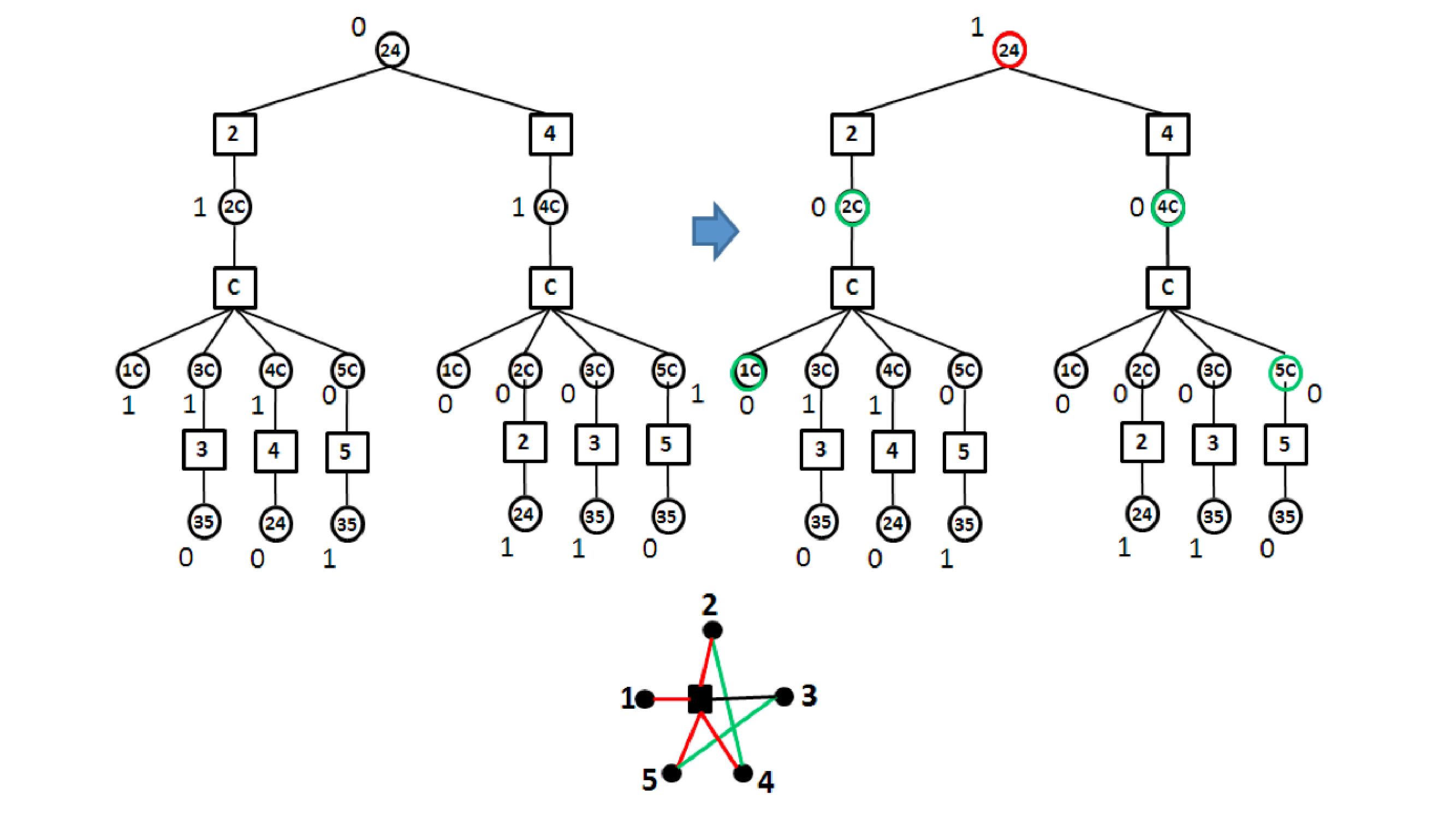}
\end{center}
\caption{\small Example of $y^{\mbox{\tiny TMAP}}$ (left) and $y^{\mbox{\tiny NEW}}$ (right), where $y^*_{1,i_C}=0$, $y^*_{2,i_C}=0$, $y^*_{3,i_C}=0$, $y^*_{4,i_C}=0$, $y^*_{5,i_C}=0$, $y^*_{3,5}=1$ and $y^*_{2,4}=1$.}
\label{fig:computation_tree}
\end{figure*}
Now we state the following lemma, playing a pivotal role for construction of $y^{\mbox{\tiny NEW}}$.
\begin{lemma}\label{lem:valid}
{\em ERROR} is never reported in the construction described above.
\end{lemma}
\begin{proof}
Proving the case when $\psi=\psi_i$ in Step 3 is easy/straightforward, therefore we only discuss here the proof for the case when $\psi=\psi_C$. Let us also assume that $y^{\mbox{\tiny NEW}}_e$ is flipped as $1\to 0$ (i.e., $y_e^*=0$), where the proof for the case $0\to 1$ follows in a similar manner. First, one can observe that $y$ satisfies $\psi_C$ if and only if $y$ is the 0-1 indicator vector of a union of disjoint even paths in the cycle $C$. Since $y^{\mbox{\tiny NEW}}_e$ is flipped as $1\to 0$, an even path including $e$ is broken into an even (possibly, empty) path and an odd (always, non-empty) path. 
We consider two cases: (a) when there exists $e^{\prime}$ within the odd path (i.e., $y^{\mbox{\tiny NEW}}_{e^{\prime}}=1$) such that $y_{e^{\prime}}^{*}=0$ and flipping
$y^{\mbox{\tiny NEW}}_{e^{\prime}}$ with $1\to0$ brakes the odd path into two even (disjoint) paths;
(b) there exists no such $e^{\prime}$ within the odd path.

In the case (a) it is easy to see that we can maintain the structure of disjoint even paths in $y^{\mbox{\tiny NEW}}$ after flipping $y^{\mbox{\tiny NEW}}_{e^{\prime}}$ as $1\to 0$, i.e.,
$\psi$ is satisfied by ${\tt FLIP}_{e^{\prime}}(y^{\mbox{\tiny NEW}})$. In the case (b), let us choose $e^{\prime}$ as a neighbor of the farthest end point (from $e$) in the odd path, i.e., $y^{\mbox{\tiny NEW}}_{e^{\prime}}=0$ (before flipping). Then, $y_{e^{\prime}}^{*}=1$ since $y^*$ satisfies factor $\psi_C$ and induces a union of disjoint even paths in the cycle $C$. Therefore, if $y^{\mbox{\tiny NEW}}_{e^{\prime}}$ is flipped as $0\to 1$, one can still maintain the structure of disjoint even paths in $y^{\mbox{\tiny NEW}}$, where $\psi$ is satisfied by ${\tt FLIP}_{e^{\prime}}(y^{\mbox{\tiny NEW}})$. The proof for the case of the $\psi$ satisfied by ${\tt FLIP}_{e^{\prime}}({\tt FLIP}_{e}(y^{*}))$ is similar. This completes the proof of Lemma \ref{lem:valid}.
\end{proof}

Due to the construction of $y^{\mbox{\tiny NEW}}$, it is a valid configuration, i.e., it satisfies all the factor functions in $T_e(t)$. 
Now, for $(i,j)\in E^{\prime}$, let $n_{ij}^{0\to 1}$ and $n_{ij}^{1\to 0}$ be the number of flip operations $0\to 1$ and $1\to 0$ for copies of $(i,j)$ in Step 2
of the construction of $y^{\mbox{\tiny NEW}}$. Then, one derives
$$w^{\prime} \cdot y^{\mbox{\tiny NEW}} =w^{\prime} \cdot y^{\mbox{\tiny TMAP}} +w^{\prime} \cdot n^{0\to 1} - w^{\prime} \cdot n^{1\to 0},$$
where $n^{0\to 1}=[n_{ij}^{0\to 1}]$ and $n^{1\to 0}=[n_{ij}^{1\to 0}]$.
We consider two cases: (i) when the path $P$ does not arrive at a leaf variable of $T_e(t)$, and (ii) otherwise.
Note that the case (i) is possible only when the condition in Step 3-1 holds during the construction of $y^{\mbox{\tiny NEW}}$. In both cases, we will show
$w^{\prime} \cdot y^{\mbox{\tiny NEW}} > w^{\prime} \cdot y^{\mbox{\tiny TMAP}}$.

{\bf Case (i).}
In this case, we define 
$y^{\dagger}_{ij}:=y^*_{ij} +\varepsilon (n_{ij}^{1\to 0}-n_{ij}^{0\to 1}),$
and establish the following lemma.
\begin{lemma}\label{lem:feasible-1}
$y^{\dagger}$ is a feasible solution of C-LP$^{\prime}$ if $\varepsilon >0$ is sufficiently small.
\end{lemma}
\begin{proof}
We have to show that $y^{\dagger}$ satisfies \eqref{eq1:clp'} and \eqref{eq2:clp'}.
Here, we prove that $y^{\dagger}$ satisfies \eqref{eq2:clp'} when $\varepsilon>0$ is sufficiently small. Then
the proof for \eqref{eq1:clp'} follows from a similar consideration. To this end, for given $C\in \mathcal C$, we consider the following polytope $\mathcal P_C$ :
$$\sum_{j\in V(C)} y_{i_C,j} \leq |C|-1,\quad y_{i_C,j}\in[0,1],\quad\forall j\in C,
$$
$$\sum_{j\in V(C)} (-1)^{d_C(j,e)} y_{i_C,j}\in [0,2],\quad\forall e\in E(C).$$
One has to show that $y^{\dagger}_C=[y_e:e\in\delta(i_C)]$ is within the polytope. It is easy to see that the condition of Step 3-1 never holds if $\psi=\psi_C$ in Step 3.
For the $i$-th copy of $\psi_C$ in $P \cap T_e(t)$, we set $y^*_C(i) = {\tt FLIP}_{e^{\prime}}({\tt FLIP}_{e}(y^{*}_C))$ in Step 3-2,
where $y^*_C(i)\in \mathcal P_C$. Since the path $P$ does not end at a leaf variable of $T_e(t)$, one finds that
$$\frac1N \sum_{i=1}^N y^*_C(i) = y^*_C + \frac1N \left(n_{C}^{1\to 0}-n_{C}^{0\to 1}\right),$$
where $N$ is the number of copies of $\psi_C$ in $P \cap T_e(t)$.
Furthermore, $\frac1N \sum_{i=1}^N y^*_C(i)\in \mathcal P_C$ due to $y^*_C(i)\in \mathcal P_C$. Therefore,
$y^{\dagger}_C\in \mathcal P_C$ if $\varepsilon \leq 1/N$. This completes the proof of Lemma \ref{lem:feasible-1}.
\end{proof}
The above lemma with $w^{\prime} \cdot y^*> w^{\prime} \cdot y^{\dagger}$ (due to the uniqueness of $y^*$) implies that
$w^{\prime} \cdot n^{0\to 1} > w^{\prime} \cdot n^{1\to 0}$, which leads to $w^{\prime} \cdot y^{\mbox{\tiny NEW}} > w^{\prime} \cdot y^{\mbox{\tiny TMAP}}$.

{\bf Case (ii).} Let us consider the case when
only one end of $P$ ends at a leaf variable $Y_e$ of $T_e(t)$.
Similar arguments applies to prove the other case
when both ends of $P$ end at leaves. 
We first define 
$y^{\ddagger}_{ij}:=y^*_{ij} +\varepsilon (m_{ij}^{1\to 0}-m_{ij}^{0\to 1}),$
where $m^{1\to 0}=[m_{ij}^{1\to 0}]$ and $m^{0\to 1}=[m_{ij}^{0\to 1}]$ is constructed as follows:
\begin{itemize}
	\item[1.] Initially, set $m^{1\to 0}, m^{0\to 1}$ by $n^{1\to 0}, n^{0\to 1}$.

\item[2.] If $y_e^{\mbox{\tiny NEW}}$ is flipped as $1\to 0$ and it is associated to a cycle parent factor $\psi_C$ for some $C\in \mathcal C$, then
decrease $m^{1\to 0}_e$ by 1 and
\begin{itemize}
\item[2-1.] If the parent $y_{e^{\prime}}^{\mbox{\tiny NEW}}$ is flipped from $1\to 0$, then decrease $m_{e^{\prime}}^{1\to0}$ by 1.
\item[2-2.] Else if there exists a `brother' edge $e^{\prime\prime}\in \delta(i_C)$ of $e$ such that $y^*_{e^{\prime\prime}}=1$ and
$\psi_C$ is satisfied by ${\tt FLIP}_{e^{\prime\prime}}({\tt FLIP}_{e^{\prime}}(y^{*}))$, then increase
$m_{e^{\prime\prime}}^{0\to 1}$ by 1.
\item[2-3.] Otherwise, report ERROR.
\end{itemize}
\item[3.] If $y_e^{\mbox{\tiny NEW}}$ is flipped as $1\to 0$ and it is associated to a vertex parent factor $\psi_i$ for some $i\in V$, then
decrease $m^{1\to 0}_e$ by 1.

\item[4.] If $y_e^{\mbox{\tiny NEW}}$ is flipped as $0\to 1$ and it is associated to a vertex parent factor $\psi_i$ for some $i\in V$, then
decrease $m^{0\to 1}_e,m^{1\to 0}_{e^{\prime}}$ by 1,
where $e^{\prime}\in \delta(i)$ is the `parent' edge of $e$, and
\begin{itemize}
\item[4-1.] If the parent $y_{e^{\prime}}^{\mbox{\tiny NEW}}$ is associated to a cycle parent factor $\psi_C$,
\begin{itemize}
\item[4-1-1.] If the grad-parent $y_{e^{\prime\prime}}^{\mbox{\tiny NEW}}$  is flipped from $1\to 0$, then decrease $m_{e^{\prime\prime}}^{1\to0}$ by 1.
\item[4-1-2.] Else if there exists a `brother' edge $e^{\prime\prime\prime}\in \delta(i_C)$ of $e^{\prime}$ such that $y^*_{e^{\prime\prime\prime}}=1$ and
$\psi_C$ is satisfied by ${\tt FLIP}_{e^{\prime\prime\prime}}({\tt FLIP}_{e^{\prime\prime}}(y^{*}))$, then increase
$m_{e^{\prime\prime\prime}}^{0\to 1}$ by 1.
\item[4-1-3.] Otherwise, report ERROR.
\end{itemize}
\item[4-2.] Otherwise, do nothing.
\end{itemize}
\end{itemize}
We establish the following lemmas.
\begin{lemma}\label{lem:valid-2}
{\em ERROR} is never reported in above construction.
\end{lemma}
\begin{lemma}\label{lem:feasible-2}
$y^{\ddagger}$ is a feasible solution of C-LP$^{\prime}$ if $\varepsilon >0$ is sufficiently small.
\end{lemma}
We omit the proofs of Lemma \ref{lem:valid-2} and Lemma \ref{lem:feasible-2} as
they are analogous to those of Lemma \ref{lem:valid} and Lemma \ref{lem:feasible-1}, respectively.
From Lemma \ref{lem:feasible-2}, one deduces
\begin{align*}
c^\prime&\leq \frac{w^{\prime}\cdot(y^*-y^{\ddagger})}{|y^*-y^{\ddagger}|}\\
&\leq
\frac{\varepsilon \left(w^{\prime}\cdot (m^{0\to 1}-m^{1\to 0})\right)}{\varepsilon (t-3)}\\
&\leq\frac{\varepsilon \left(w^{\prime}\cdot (n^{0\to 1}-n^{1\to 0})+3w^{\prime}_{\max}\right)}{\varepsilon (t-3)},
\end{align*}
where $|y^*-y^{\ddagger}|\geq \varepsilon(t-3)$ follows from the fact that $P$ hits a leaf variable of $T_e(t)$ and there are at most three
increases or decreases in $m^{0\to 1}$ and $m^{1\to 0}$ in the above  construction.
Hence,
$w^{\prime} \cdot (n^{0\to 1}-n^{1\to 0})\geq c^\prime(t-3) - 3w^{\prime}_{\max} >0$ 
if $t>\frac{3w^{\prime}_{\max}}{c^\prime}+3 = O\left(\frac{w^{\prime}_{\max}}{c^\prime}\right),$
which implies $w^{\prime} \cdot y^{\mbox{\tiny NEW}} > w^{\prime} \cdot y^{\mbox{\tiny TMAP}}$. 
This completes the proof of Lemma \ref{thm3}.

\section{Cutting-plane Algorithm using Belief Propagation}\label{sec:cutting}

As established in the preceding section, BP on a carefully designed GM, built with  appropriate odd cycle extensions,
solves the MWM problem as long as the corresponding MWM-LP relaxation is tight. 
However, finding a proper collection of odd cycles to ensure tightness of the MWM-LP with respect to
given weights is a challenging task. 
In this section, we provide a heuristic, coined Cutting-Plane method using BP (CP-BP), for tackling this challenging task. 

\subsection{CP-BP using Half-Integrality}

CP-BP consists of making sequential, ``cutting plane'', modifications to the underlying GM using the output of the BP algorithm in the previous step. CP-BP is defined as follows:
\begin{itemize}
	\item[1.] Initialize $\mathcal{C} = \emptyset$.
	\item[2.] Run BP on C-GM$^\prime$ for $T$ iterations, where we choose $T=500$ in all our experiments.
	\item[3.] For each edge $e\in E$, set
	\begin{equation*}
	   y_e=\begin{cases}1&\mbox{if}~n^T_e[1]>n^T_e[0]~\mbox{and}~n^{T-1}_e[1]>n^{T-1}_e[0]\\
	0&\mbox{if}~n^T_e[1]<n^T_e[0]~\mbox{and}~n^{T-1}_e[1]<n^{T-1}_e[0]\\
	1/2&\mbox{otherwise}
	\end{cases}.
	\end{equation*}
	\item[4.] Compute $x=[x_e]$ using $y=[y_e]$ as per \eqref{eq:mainthm}, and terminate if $x\notin \{0,1/2,1\}^{|E|}$.
	\item[5.] If there is no edge $e$ with $x_e=1/2$, return $x$.
	Else if a non-intersecting odd cycle of edges $\{e:x_e=1/2\}$ exists, it is added to $\mathcal C$ and go to Step $2$. Otherwise, terminate.
\end{itemize}
Notice that in this procedure, BP can be replaced by an LP solver to obtain $x$ in Step 4. This results in a traditional cutting-plane LP (CP-LP) method for the MWM problem \cite{GH85}. The primary reason why we design CP-BP to terminate when $x\notin \{0,1/2,1\}^{|E|}$ is because
the solution $x$ of C-LP is always half-integral
as shown in the following section.
Note that $x\notin \{0,1/2,1\}^{|E|}$ occurs in CP-BP when BP does not find the solution of C-LP.
In our experiments, we observe that it often occurs primarily
because BP is only guaranteed to find an integral (instead of half-integral) 
solution of C-LP.

\begin{figure*}[t!]
\begin{center}
\hspace{0.1in}
\includegraphics[width=0.48\textwidth]{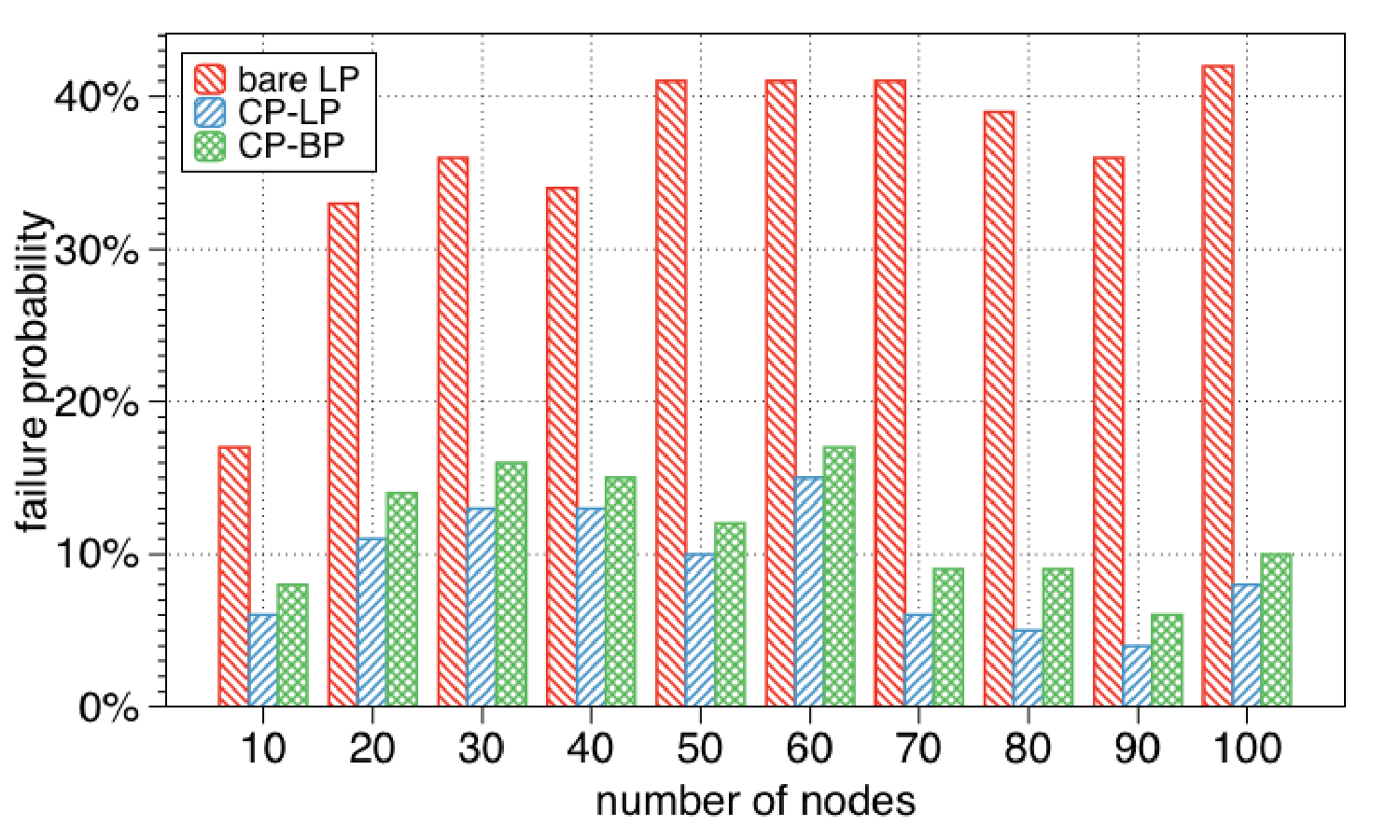}
\hfill
\includegraphics[width=0.48\textwidth]{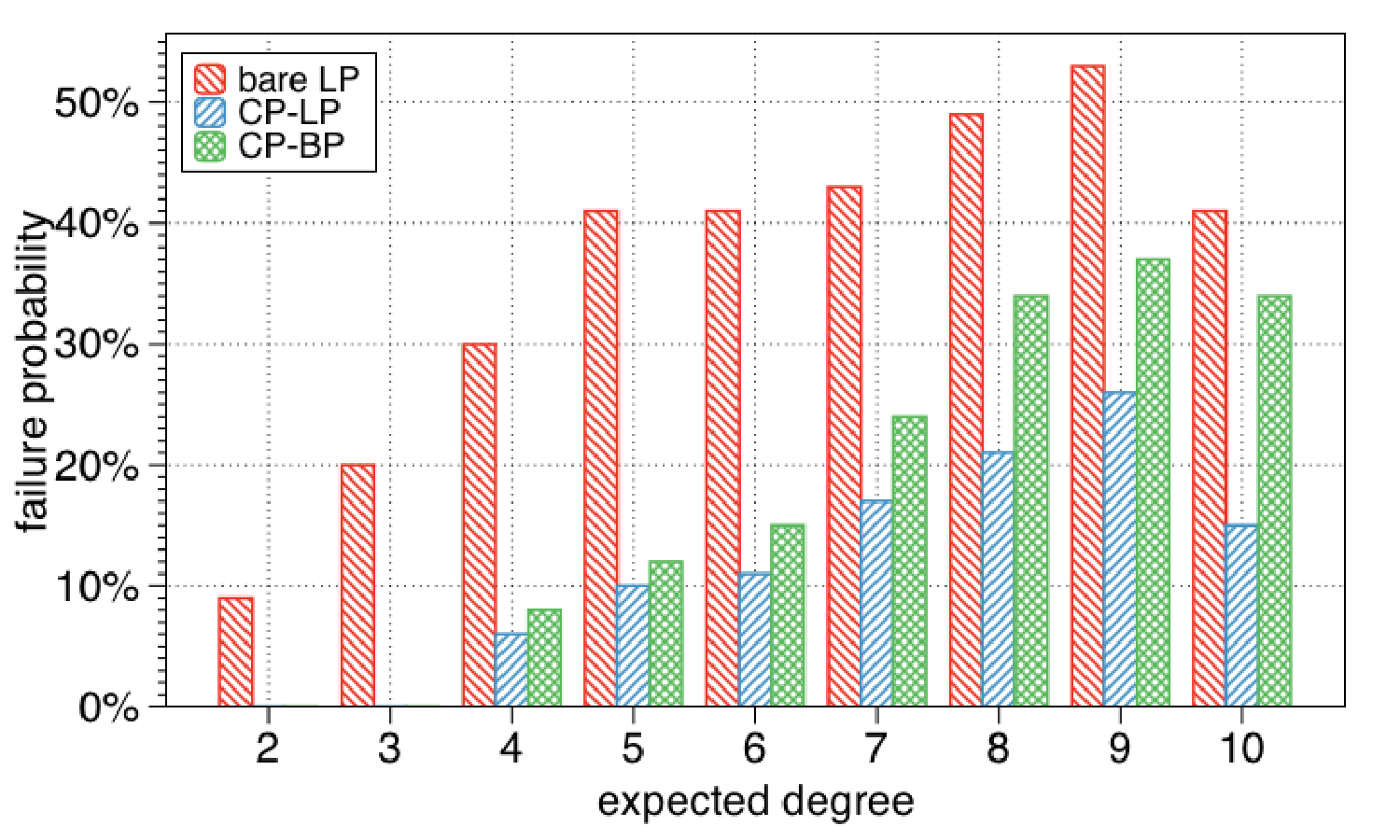}
\hspace{0.1in}
\end{center}
\vspace{-0.4in}
\caption{Evaluation of CP-BP and CP-LP on random MWM instances with varying number of nodes (left) and expected degree of vertices (right). 
The bare LP corresponds to the initial MWM-LP without cycle constraints, i.e., $\mathcal{C}=\emptyset$. 
Each plot is averaged over $100$ instances.
}
\label{fig:experiment}
\end{figure*}

In the following experiments we compare CP-BP and CP-LP in order to gauge the effectiveness of BP as an LP solver for MWM problems - i.e., to test if BP is as powerful as an LP solver over the class of problems considered. 
To this end, we create each 
random MWM instances on Erd\H{o}s-Renyi (ER) graphs, i.e.,
sparsify a complete graph 
by eliminating edges with fixed probability $p\in (0,1)$ independently 
then assign integral weights, 
drawn uniformly from $[1,10^{6}]$ independently, to the remaining edges. 
For the performance measure, we compare the ratio (i.e., percentage) of MWM instances where the respective algorithm fails to find the exact solution. 
In the first experiment, we vary the number of nodes $|V|$ while fixing the average degree of the graph to $5$, i.e., $(1-p)(|V|-1)=5$. 
In the second experiment, we fix the number of nodes to $50$, and control the sparsity of the ER graph by varying the expected degree of each nodes, i.e., $(1-p)(|V|-1)$.
Our experimental results are summarized in Figure~\ref{fig:experiment}. The results show that 
BP is quite effective at finding even half-integral solutions of C-LP as
CP-BP is almost as good as CP-LP for solving the MWM problem:
our graphical transformation allows BP to solve significantly more MWM problems than are solvable by bare 
LP (or BP) run without additional cutting-plane schemes. 
\subsection{Proof for Half-Integrality of C-LP}
In this section, we show that there always exists a half-integral solution of C-LP.
To this end, it suffices to show that every vertex in the constraint polytope of C-LP is half-integral.
Let $x\notin\{0,1/2,1\}^{|V|}$ be a feasible point in the constraint polytope of C-LP, and define the following notation.
\begin{eqnarray*}
E_x& =& \left\{ e\in E: x_e\in (0,1) \right\}\\
V_x &=& \big\{v\in V: \sum_{e\in\delta(v)} x_e=1\big\}\\
\mathcal C_x &=& \bigg\{C\in \mathcal C : \sum_{e\in E(C)} x_e=\frac{|C|-1}2\bigg\}
\end{eqnarray*}
Namely, $E_x$ is the set of non-integral edges and $V_x$ (or $\mathcal C_x$)
is the set of `tight' vertices (or cycles).
Our goal is to show that $x$ is not a vertex of the polytope.
We first state the following key lemma.
\begin{lemma}\label{lem:cx}
For every $C\in\mathcal{C}_x$ such that $E(C)\cap E_x\ne\emptyset$, the following statements hold:
\begin{itemize}
\item[(a)] If there exists $(u,v)\in C$ such that $x_{(u,v)}=0$, then every $w\in V(C)\setminus\{u,v\}$
having  
an odd graph distance from $v$ (or $u$) satisfies \begin{equation*}
\sum_{e\in\delta(w)\cap E(C)}x_{e}=1.
\end{equation*}
\item[(b)] Otherwise, there exist at least two vertices $v_1,v_2\in V(C)$ such that \begin{equation*}
    \sum_{e\in\delta(v_i)\cap E(C)}x_{e}<1
\end{equation*} for all $i\in\{1,2\}$. In particular, $x_e\in E_x$ for all $e\in E(C)$.
\end{itemize}
\end{lemma}
\begin{proof}
One can observe that (a) is trivial to hold as $C\in\mathcal C_x$.
Now consider the case (b), i.e., there does not exist zero edge in $C$.
If $\sum_{e\in\delta(v)\cap E(C)}x_{e}=1$ for all $v\in V(C)$, then 
\begin{equation*}
\sum_{e\in E(C)}x_e=\sum_{v\in V(C)}\sum_{e\in\delta(v)\cap E(C)}x_{e}=|C|/2,
\end{equation*}
which violates the assumption that $x$ is a feasible point in the constraint polytope of C-LP.
Similarly, one can argue the case when only one vertex $v\in V(C)$ satisfies $\sum_{e\in\delta(v)\cap E(C)}x_{e}<1$.
This completes the proof of (b). 
\end{proof}

In what follows, we will show that $x$ is not a vertex, i.e.,
there exist two different feasible points $y=[y_e],z=[z_e]$ satisfying $x = (y+z)/2$. To this end, we introduce the following lemma.
\begin{lemma}\label{lem:p}
There exists a walk $\mathcal W\ne\emptyset$ along $E_x$ satisfying the followings:
\begin{itemize}
\item[(a)] There exists $e\in E(\mathcal W)$ such that $x_e\notin\{0,1/2,1\}$.
\item[(b)] $\mathcal W=\mathcal T_1\cup\mathcal T_2\cup\mathcal B$ where $\mathcal T_1,\mathcal T_2$ are closed walks and $\mathcal B$ is a path connecting them. 
Here, $\mathcal T_1$, $\mathcal T_2$ or $\mathcal B$ can be $\emptyset$.
\item[(c)] For all $\mathcal S\in\{\mathcal T_1,\mathcal T_2,\mathcal B\}$ and for all $C\in \mathcal C_x$, $|E(\mathcal S)\cap E(C)|$ is even.
\item[(d)] If $\mathcal W$ is a single 
closed walk, i.e., $\mathcal T_2=\mathcal B=\emptyset$, then either
$\sum_{e\in\delta(v)\cap E(\mathcal W)}x_e=1$ for all $v\in V(\mathcal W)$ or there exists a vertex $v\in V(\mathcal W)$ such that $\sum_{e\in\delta(v)}x_e<1$.
\item[(e)] For all $v\in V(\mathcal W)$ such that $|\delta(v)\cap E(\mathcal W)|=1$
(i.e., $v$ is an end of walk $\mathcal W$), $\sum_{e\in \delta(v)}x_e<1.$
\end{itemize}
\end{lemma}
Due to (b) in Lemma \ref{lem:p}, we split the cases that either $\mathcal{W}$ is a path (i.e., $\mathcal T_1=\mathcal T_2=\emptyset$) or $\mathcal W$ contains a closed walk.

{\bf $\mathcal W$ is a path.}
In this case, we construct  $y=[y_e],z=[z_e]$ satisfying $x = (y+z)/2$ by starting from $x=[x_e]$ and alternatively adding/subtracting some constant $\varepsilon>0$
following the path $\mathcal W= e_1\rightarrow e_2 \rightarrow \cdots$: $y_e=z_e=x_e$ for $e\notin \mathcal W$ and
\begin{align*}
y_{e_1}\leftarrow x_{e_1} +\varepsilon,\quad y_{e_2}\leftarrow x_{e_2} -\varepsilon,\quad y_{e_3}\leftarrow x_{e_3} +\varepsilon,\cdots\\
z_{e_1}\leftarrow x_{e_1} -\varepsilon,\quad z_{e_2}\leftarrow x_{e_2} +\varepsilon,\quad z_{e_3}\leftarrow x_{e_3} -\varepsilon,\cdots
\end{align*}
We provide an example of $\{x,y,z\}$ in Figure~\ref{fig:case_a}.
Due to (c) and (e) in Lemma \ref{lem:p},
one can check that $y,z$ are feasible points to the constraint polytope of C-LP with sufficiently small $\varepsilon>0$.

\begin{figure*}[ht]
\begin{center}
\includegraphics[width=0.9\textwidth]{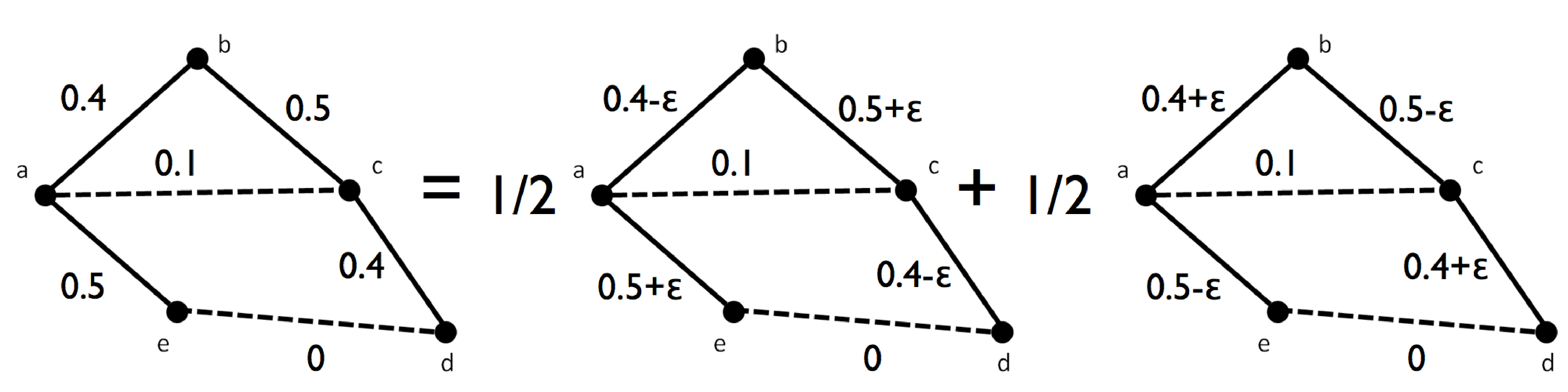}
\end{center}
\caption{\small Example of $x$ (left), $y$ (middle) and $z$ (right) for the case that $\mathcal W$ is a path, where $\mathcal W = e\to a\to b \to c \to d$.}
\label{fig:case_a}
\end{figure*}

{\bf $\mathcal W$ contains a closed walk.} 
In this case, if the length of a closed walk $\mathcal T_1$ or $\mathcal T_2$ in $\mathcal W$ is even, one can
construct two different feasible points $y=[y_e],z=[z_e]$ with $x = (y+z)/2$ by starting from $x=[x_e]$ and alternatively adding/subtracting some
small constant $\varepsilon>0$ following the cycle, similar to the case when $\mathcal W$ is a path.
Therefore, one can assume that $\mathcal W$ is one of the following cases:
\begin{itemize}
\item[(i)] $\mathcal W$ is an odd closed walk, i.e., $\mathcal T_2=\mathcal B=\emptyset$.
\item[(ii)] $\mathcal W$ is a union of two odd closed walks and a path connecting them, i.e.,
$\mathcal T_1,\mathcal T_2\neq \emptyset$ (note that $\mathcal B=\emptyset$ is possible).
\item[(iii)] $\mathcal W$ is a union of an odd closed walk and a path connected to it, i.e., $\mathcal T_2=\emptyset$.
\end{itemize}
First, consider Case (i).
From (d) of Lemma \ref{lem:cx}, either
$\sum_{e\in\delta(v)\cap E(\mathcal W)}x_e=1$ for all $v\in V(\mathcal W)$ or there exists a vertex $v\in V(\mathcal W)$ such that $\sum_{e\in\delta(v)}x_e<1$.
However, $\sum_{e\in\delta(v)\cap E(\mathcal W)}x_e=1$ for all $v\in V(\mathcal W)$ is impossible because in this case,
every $x_e$ for $e\in \mathcal P$ should be half-integral and this
contradicts to (a) of Lemma \ref{lem:p}.
In another case, one can construct two different feasible points $y=[y_e],z=[z_e]$ with $x = (y+z)/2$ by starting from $v$ and alternatively adding/subtracting some
small constant $\varepsilon>0$ following the cycle, similar to the case when $\mathcal W$ is a path.

Next, we consider Case (ii)
where 
Case (iii) can be handled in a similar manner.
We construct two different feasible points $y=[y_e],z=[z_e]$ with $x = (y+z)/2$ by starting from $x=[x_e]$ and
\begin{itemize}
\item alternatively adding/subtracting $\varepsilon$ following closed walks $\mathcal T_1$ and $\mathcal T_2$
\item alternatively adding/subtracting $2\varepsilon$ following the bridge edges $\mathcal B$
\end{itemize}
where $\varepsilon>0$ is small enough and
the feasibility of $y,z$ is from (c) of Lemma \ref{lem:p}.
We provide an example of $\{x,y,z\}$ in Figure~\ref{fig:case_b}. This implies that $x$ is not a vertex, thus completing the proof of the half-integrality of C-LP.

\begin{figure*}[ht]
\begin{center}
\includegraphics[width=1\textwidth]{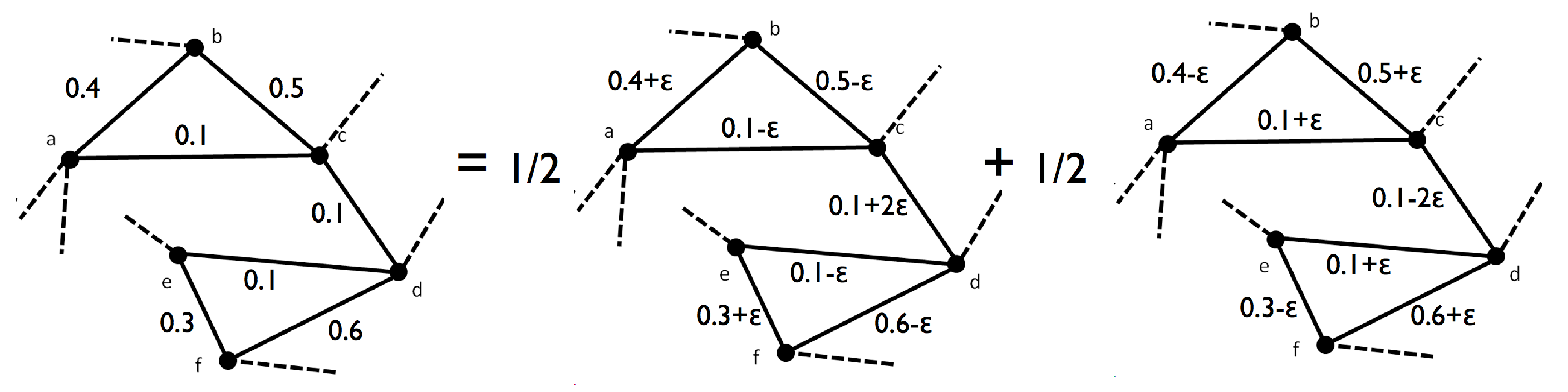}
\end{center}
\caption{\small Example of $x$ (left), $y$ (middle) and $z$ (right) for the case that $\mathcal W$ consists of two cycles and a path between them, where $\mathcal T_1 = a\to b \to c\to a$, $\mathcal B=c\to d$
and $\mathcal T_2
= e\to d\to f\to e$.}
\label{fig:case_b}
\end{figure*}

\subsection{Proof of Lemma \ref{lem:p}}
We plan to explicitly construct a walk $\mathcal W$ consisting of edges in $E_x$ as follows.
If there exists $e\in E_x\setminus\bigcup_{C\in\mathcal C_x}E(C)$ such that $x_{e}\ne 1/2$, initialize $\mathcal W$ by $e$.
Otherwise,
pick $C\in\mathcal C_x$ such that $E(C)$ contains a non-half-integral edge:
such $C$ always exists as we assumed $x\notin\{0,1/2,1\}^{|V|}$.
Furthermore, by Lemma \ref{lem:cx}, there always exist $u,v\in V(C)$ such that 
\begin{equation*}\sum_{e^\prime \in\delta(u)\cap E(C)}x_e<1
\qquad \sum_{e^\prime\in\delta(v)\cap E(C)}x_e<1
\end{equation*}
and the even path along $E(C)\cap E_x$ between $u,v$ contains a
non-half-integral edge. 
Then, we initialize $\mathcal W$ by the even path.
We expand $\mathcal W$ in both directions using the following procedure for each ending edge $e^{\mbox{\tiny END}}$ and the corresponding ending vertex $v^{\mbox{\tiny END}}$ of $\mathcal W$:
\begin{itemize}
\item[0.] Initially, find $C_{e^{\mbox{\tiny END}}}\in\mathcal C_x$ such that $e^{\mbox{\tiny END}}\in C_{e^{\mbox{\tiny END}}}$.
If such $C_{e^{\mbox{\tiny END}}}$ does not exist, set $C_{e^{\mbox{\tiny END}}}=\emptyset$.
\item[1.] Repeat the followings:
\begin{itemize}
\item[1-1.] If $v^{\mbox{\tiny END}}$ is not a vertex of any cycle in $\mathcal C_x\setminus\{C_{e^{\mbox{\tiny END}}}\}$, 
then choose $e \in(\delta(v^{\mbox{\tiny END}})\cap E_x)\setminus  (E(\mathcal W)\cup E(C_{e^{\mbox{\tiny END}}}))$ and add $e$ to $\mathcal W$. 
Set $e^{\mbox{\tiny END}}\leftarrow e$. If such $e$ does not exist, terminate Step 1.
\item[1-2.] 
Else if $v^{\mbox{\tiny END}}$ is a vertex of some cycle $\mathcal C_{v^{\mbox{\tiny END}}} \in C_x\setminus\{C_{e^{\mbox{\tiny END}}}\}$, then
 there exists $u\in V(C_{v^{\mbox{\tiny END}}})$ such that \begin{equation*}\sum_{e\in\delta(u)\cap E\left(C_{v^{\mbox{\tiny END}}}\right)}x_{e}<1
 \qquad\mbox{from Lemma \ref{lem:cx}}.
 \end{equation*}
 Lemma \ref{lem:cx} also implies that there exists an even path along $E(C_{v^{\mbox{\tiny END}}})\cap E_x$ from $v^{\mbox{\tiny END}}$ to $u$.
Add the even path between $u,v^{\mbox{\tiny END}}$ to $\mathcal W$ and set $e^{\mbox{\tiny END}}$ to be the last edge of the path.
\item[1-3]
For the new ending edge $e^{\mbox{\tiny END}}$ and the new ending vertex $v^{\mbox{\tiny END}}$, 
find $C_{e^{\mbox{\tiny END}}}\in\mathcal C_x$ such that $e^{\mbox{\tiny END}}\in C_{e^{\mbox{\tiny END}}}$. If such $C_{e^{\mbox{\tiny END}}}$ does not exist, set $C_{e^{\mbox{\tiny END}}}=\emptyset$.
\end{itemize}
until one of the following events occur: 
\begin{itemize}
\item[1-4] $e^{\mbox{\tiny END}}$ touches $\mathcal{W}$ but does not touch any odd cycle in $\mathcal C_x$ intersecting with $\mathcal W$, i.e., $v^{\mbox{\tiny END}}$ was already in $V(\mathcal W)$ before $e^{\mbox{\tiny END}}$ is added and $v^{\mbox{\tiny END}}\notin V(C)$ for all $C\in\mathcal C_x\setminus\{C_{e^{\mbox{\tiny END}}}\}$. 
\item[1-5] $e^{\mbox{\tiny END}}$ touches some odd cycle in $\mathcal C_x$ intersecting with $\mathcal W$, i.e, $v^{\mbox{\tiny END}}\in V(C)$ for some $C\in\mathcal C_x\setminus\{C_{e^{\mbox{\tiny END}}}\}$ such that  $E(C)\cap E(\mathcal W)\ne\emptyset$. Then, $\mathcal W$ contains an even path $\mathcal P=v\to\cdots\to u$ in $C$ due to Step 1-2. Without loss of generality, assume that $u$ was added to $\mathcal W$ after $v$.
If $v^{\mbox{\tiny END}}\in V(\mathcal P)$, continue the expansion from $v^{\mbox{\tiny END}}$ to $v$ by following an even path in $E(C)\cap E_x$.
Otherwise, continue the expansion from $v^{\mbox{\tiny END}}$ to $u$ by following an even path in $E(C)\cap E_x$.
\end{itemize}
\item[2.] If $\mathcal W$ is a closed walk but $V(\mathcal W)\not\subset V_x$, then choose $v\in V(\mathcal W)\setminus V_x$ and go to Step 1 with ending vertex $v^{\mbox{\tiny END}}\leftarrow v$. Otherwise, terminate.
\end{itemize}
We note that when the procedure terminates for one end of $\mathcal W$, the procedure for the other end continues until it meets the termination criteria.
Now, we verify (a)-(e) in Lemma \ref{lem:p} with the constructed walk $\mathcal W$ as follows:
\begin{itemize}
\item[(a)] As we started from an edge/walk containing $x_{e^*}\in\{0,1/2,1\}$, (a) holds.
\item[(b)] The termination criteria in Step 1 is either when a closed walk appears (Steps 1-4 and 1-5) or the walk ends without
being closed (Step 1-1). 
Hence, (b) holds if in Step 1-5, one chooses $\mathcal T_i$ as the closed walk from $u$ to $u$ or $v$ to $v$ corresponding to where the walk ends. 
\item[(c)] We design Step 1-5 for (c). 
\item[(d)] Step 2 and the termination condition in Step 1-1 implies (d).
\item[(e)] This is because
$v\in V(\mathcal W)$ such that $|\delta(v)\cap E(\mathcal W)|=1$ only appears at the termination in Step 1-1.
\end{itemize}
This completes the proof of Lemma \ref{lem:p}.

\section{Conclusion}
\label{sec:concl}

It was recently shown that BP converges to the correct MAP assignment for a certain class of natural loopy GM formulations of classical combinatorial optimization problems. While most of the prior work on this subject has been focused on `BP analysis', in this paper we choose another path. We explore the freedom in adding to the original GM new constraints that on one hand keep the MAP computation invariant (the same) and on the other hand allows to certify convergence and correctness of BP. 
Even though the technique of this paper was developed solely for the MWM model, we believe that our approach is extendable to a broader set of MAP problems.
Our approach relays on finding a sequence of additional constraints such that the respective LP relaxation is gapless (in terms of its use to compute the MAP assignment). Therefore, another important direction for future analytic work is related to developing a provably efficient and distributed (or semi-distributed) way of discovering the sequence of constraints (odd cycles within the MWM model) closing the gap between respective LP and IP (Integer Programming).

\subsubsection*{Acknowledgment}

The work at LANL was carried out under the auspices of the National Nuclear Security Administration of the U.S. Department of Energy at Los Alamos National Laboratory under Contract No. DE-AC52-06NA25396.


\begin{thebibliography}{1}
\bibitem{05YFW}
J.~Yedidia, W.~Freeman, and Y.~Weiss, ``Constructing free-energy approximations
  and generalized belief propagation algorithms,'' \emph{IEEE Transactions on Information Theory}, vol.~51, no.~7, pp.~2282 -- 2312, 2005.

\bibitem{08RU}
T.~J. Richardson and R.~L. Urbanke, \emph{Modern Coding Theory.}\hskip 1em plus
  0.5em minus 0.4em\relax Cambridge University Press, 2008.

\bibitem{09MM}
M.~Mezard and A.~Montanari, \emph{Information, physics, and computation}, ser.
  Oxford Graduate Texts.\hskip 1em plus 0.5em minus 0.4em\relax Oxford: Oxford
  Univ. Press, 2009.

\bibitem{08WJ}
M.~J. Wainwright and M.~I. Jordan, ``Graphical models, exponential families,
  and variational inference,'' \emph{Foundations and Trends in Machine
  Learning}, vol.~1, no.~1, pp.~1--305, 2008.

\bibitem{09GLG}
J.~Gonzalez, Y.~Low, and C.~Guestrin.
``Residual splash for optimally parallelizing belief propagation,"
in \emph{International Conference on Artificial Intelligence and Statistics}, 2009.

\bibitem{10LGKBGH}
Y.~Low, J.~Gonzalez, A.~Kyrola, D.~Bickson, C.~Guestrin, and J.~M.~Hellerstein,
``GraphLab: A New Parallel Framework for Machine Learning,"
in \emph{Uncertainty in Artificial Intelligence (UAI)}, 2010.

\bibitem{08BSS}
M.~Bayati, D.~Shah, and M.~Sharma, ``Max-product for maximum weight matching:
  Convergence, correctness, and lp duality,'' \emph{IEEE Transactions on Information Theory}, vol.~54, no.~3, pp.~1241 --1251, 2008.

\bibitem{11SMW}
S.~Sanghavi, D.~Malioutov, and A.~Willsky, ``Linear Programming Analysis of Loopy Belief Propagation for Weighted Matching,''
in \emph{Neural Information Processing Systems (NIPS)}, 2007

\bibitem{07HJ}
B.~Huang, and T.~Jebara, ``Loopy belief propagation for bipartite maximum weight b-matching,''
in \emph{Artificial Intelligence and Statistics (AISTATS)}, 2007.

\bibitem{11BBCZ}
M.~Bayati, C.~Borgs, J.~Chayes, R.~Zecchina,
``Belief-Propagation for Weighted b-Matchings on Arbitrary Graphs and its Relation to Linear Programs with Integer Solutions,''
\emph{SIAM Journal in Discrete Math}, vol.~25, pp.~989--1011, 2011.

\bibitem{08RT}
N.~Ruozzi, Nicholas, and S.~Tatikonda,
``st Paths using the min-sum algorithm,"
in \emph{46th Annual Allerton Conference on Communication, Control, and Computing}, 2008.

\bibitem{07SSW}
S.~Sanghavi, D.~Shah, and A.~Willsky, ``Message-passing for max-weight
  independent set,'' in \emph{Neural Information Processing Systems (NIPS)}, 2007.

\bibitem{10GSY}
D.~Gamarnik, D.~Shah, and Y.~Wei, ``Belief propagation for min-cost network
  flow: convergence \& correctness,'' in \emph{SODA}, pp.~279--292, 2010.

\bibitem{WJW05}
M.~J.~Wainwright, T.~Jaakkola, and A.~S.~Willsky, 
``Map estimation via agreement on trees: message passing
and linear programming,'' 
\emph{IEEE Transactions on Information Theory}, vol.~51, no.~11, pp.~1120–1146, 2005.

\bibitem{KW05}
V.~Kolmogorov and M.~J.~Wainwright,
``On the optimality of tree-reweighted max-product message passing,''
\emph{Uncertainty in Artificial Intelligence (UAI)}, 2005.

\bibitem{GJ08}
A.~Globerson and T.~Jaakkola,
``Fixing max-product: Convergent message passing algorithms for MAP LP-relaxations,'' 
\emph{Neural Information Processing Systems (NIPS}, 2008.

\bibitem{W07}
T.~Werner, 
``A linear programming approach to max-sum problem: A review,'' 
\emph{IEEE Transactions on Pattern Analysis and Machine Intelligence,} 
vol.~29, no.~7, pp.~1165-1179, 2007.

\bibitem{MGW09}
T.~Meltzer, A.~Globerson and Y.~Weiss, 
``Convergent message passing algorithms: a unifying view,'' 
\emph{Uncertainty in Artificial Intelligence (UAI)}, 2009.



\bibitem{54DFJ}
G.~Dantzig, R.~Fulkerson, and S.~Johnson, ``Solution of a large-scale traveling-salesman problem,''
\emph{Operations Research}, vol.~2, no.~4, pp.~393--410, 1954.

\bibitem{SCL12}
D.~Sontag, D.~K.~Choe and Y.~Li, 
``Efficiently Searching for Frustrated Cycles in MAP Inference.'' 
\emph{Uncertainty in Artificial Intelligence (UAI)}, 2012.

\bibitem{12CLS}
K.~Chandrasekaran, L.~A.~Vegh, and S.~Vempala.
``The cutting plane method is polynomial for perfect matchings,''
in \emph{Foundations of Computer Science (FOCS)}, 2012

\bibitem{63G}
R.~G.~Gallager, ``Low Density Parity Check Codes,''
\emph{MIT Press}, Cambridge, MA, 1963.

\bibitem{97W}
Y.~Weiss, ``Belief propagation and revision in networks with loops,''
\emph{MIT AI Laboratory}, Technical Report 1616, 1997.

\bibitem{00FK}
B.~J.~Frey, and R.~Koetter, ``Exact inference using the attenuated max-product
algorithm,'' \emph{Advanced Mean Field Methods: Theory and Practice, ed.
Manfred Opper and David Saad, MIT Press}, 2000.

\bibitem{01WF}
Y.~Weiss, and W.~T.~Freeman, ``On the Optimality of Solutions of the
Max–Product Belief–Propagation Algorithm in Arbitrary Graphs,'' \emph{IEEE
Transactions on Information Theory}, vol.~47, no.~2, pp.~736--744. 2001.

\bibitem{GH85}
M.~Grotschel, and O.~Holland, ``Solving matching problems with linear programming,''
\emph{Mathematical Programming}, vol.~33, no.~3, pp.~243--259. 1985.

\bibitem{13CGS}
M.~Chertkov, A.~Gelfand, and J.~Shin, ``Loop calculus and bootstrap-belief propagation for perfect matchings on arbitrary graphs,'' \emph{Journal of  Physics: Conference Series}, vol.~473, p.~012007. 2013.

\bibitem{PS145}
S.~Park, and J.~Shin,
``Max-Product Belief Propagation for Linear Programming: Convergence and Correctness,'' 
\emph{Uncertainty in Artificial Intelligence (UAI)}, 2015.



\bibitem{65Edm}
J.~Edmonds, ``Paths, trees, and flowers", \emph{Canadian Journal of Mathematics},
vol.~3, pp. 449--467, 1965.

\bibitem{MVV87}
K.~Mulmuley, U.~Vazirani, and V.~Vazirani, 
``Matching is as easy as matrix inversion,''
in \emph{Combinatorica}, vol.~7, pp.~105-–113, 1987.

\end{thebibliography}
\end{document}